\documentclass[journal, romanappendices, final]{IEEEtran}
\usepackage{latexsym,amsfonts,amssymb,amsthm, setspace, verbatim,cite,mathrsfs,graphicx,bm,stfloats, hyperref, tikz,xcolor,soul}
\usepackage[cmex10]{amsmath}
\newcommand*{\qedfilled}{\hfill\ensuremath{\blacksquare}}

\newcommand{\be}{\begin{equation}}
\newcommand{\ee}{\end{equation}}
\newcommand{\ben}{\begin{equation*}}
\newcommand{\een}{\end{equation*}}
\newcommand{\mc}{\mathcal}
\newtheorem{lem}{Lemma}

\newtheorem{defi}{Definition}
\newtheorem{thm}{Theorem}
\newtheorem{fact}{Fact}

\newtheorem{corr}{Corollary}
\newcounter{mytempeqncnt}
\newcommand{\e}{\epsilon}
\newcommand{\mcb}{\mathcal{B}_{M,L}}
\newcommand{\bfs}{\mathbf{S}}
\newcommand{\bfsh}{\mathbf{\hat{S}}}

\newcommand{\abs}[1]{\lvert#1\rvert}
\newcommand{\norm}[1]{\lVert#1\rVert}
\newcommand{\expec}{\mathbb{E}}

\begin{document}
\title{Lossy Compression via Sparse Linear Regression: Performance under Minimum-distance Encoding}
\author{Ramji Venkataramanan,~\IEEEmembership{Member,~IEEE,}
Antony Joseph,
and Sekhar Tatikonda,~\IEEEmembership{Senior Member,~IEEE}

\thanks{This work was partially supported by NSF Grants CCF-1017744 and CCF-1217023. The material in this paper was presented in part at the 2012 IEEE International Symposium on Information Theory.}%
\thanks{R.~Venkataramanan is with the Department of Engineering, University of Cambridge, Cambridge CB2 1PZ, UK (e-mail: ramji.v@eng.cam.ac.uk).}%
\thanks{A.~Joseph is with Department of Statistics, University of California, Berkeley, CA
94704 USA (e-mail: antony.joseph@stat.berkeley.edu).}
\thanks{S. Tatikonda is with the Department of Electrical Engineering, Yale University, New Haven CT 06511, USA (e-mail: sekhar.tatikonda@yale.edu).}
\thanks{Communicated by M.~Elad, Associate Editor for Signal Processing.}
}
\maketitle

\begin{abstract}
  We study a new class of codes for lossy compression with the squared-error distortion criterion, designed using the statistical framework of high-dimensional linear regression. Codewords are linear combinations of subsets of columns of a design matrix. Called a Sparse Superposition or Sparse Regression codebook, this structure is motivated by an analogous construction proposed recently by Barron and Joseph for communication over an AWGN channel. For i.i.d Gaussian sources and  minimum-distance encoding, we show that such a code  can attain the Shannon rate-distortion function with the optimal error exponent, for all distortions below a specified value. It is also shown that  sparse regression codes are robust in the following sense: a codebook designed to compress an i.i.d Gaussian source of variance $\sigma^2$  with (squared-error) distortion $D$ can compress any ergodic source of  variance less than $\sigma^2$ to within distortion $D$. Thus the sparse regression ensemble retains many of the good covering properties of the i.i.d random Gaussian ensemble, while having having a compact representation in terms of a matrix whose size is a low-order polynomial in the block-length.
\end{abstract}


\section{Introduction}
\label{sec:intro}
\IEEEPARstart{O}{ne} of the important outstanding problems in information theory is the development of practical codes for lossy compression of general sources at rates approaching Shannon's rate-distortion bound. In this paper, we study a class of codes called Sparse Superposition Codes or  Sparse Regression Codes (SPARC) for compression under the squared-error distortion criterion.  These codes are constructed based on the statistical framework of high-dimensional linear regression. The codewords are sparse linear combinations of columns of an $n \times N$ design matrix or `dictionary', where $n$ is the block-length and $N$ is a low-order polynomial in $n$. This codebook structure is  motivated by an analogous construction proposed recently by Barron and Joseph for communication over an AWGN channel \cite{AntonyML, AntonyFast}. The sparse regression structure enables the design of computationally efficient encoders based on the rich theory on sparse approximation and sparse signal recovery \cite{MP93, BarronCDD08, TroppOMP, CandesTaoLP, DonohoCS}.  Here, the performance of these codes under minimum-distance encoding is studied. The design of computationally feasible encoders is discussed in a companion paper \cite{RVGaussianFeasible}.

We lay down some notation before proceeding further. Upper-case letters are used to denote random variables, lower-case for their realizations,  and bold-face letters to denote random vectors and matrices.  All vectors  have length $n$. The  source sequence  is denoted by $\bfs \triangleq (S_1, \ldots, S_n)$, and the reconstruction sequence by $\bfsh \triangleq (\hat{S}_1, \ldots, \hat{S}_n)$.   The $\ell_2$-norm of vector $\mathbf{X}$ is denoted by $\norm{\mathbf{X}}$, and $\abs{\mathbf{X}} =  \norm{\mathbf{X}} / \sqrt{n}$ is the normalized version.

We use natural logarithms unless otherwise mentioned; entropy and mutual information are therefore measured in nats.   The Gaussian distribution with mean $\mu$ and variance $\sigma^2$ is denoted by  $\mc{N}(\mu,\sigma^2)$. The symbol $\kappa$ is used to denote a generic positive constant whose exact value is not needed. $f(x)=o(g(x))$ means $\lim_{x \to \infty} f(x)/g(x) =0$; $f(x)=\Theta(g(x))$ means $f(x)/g(x)$ asymptotically lies in an interval $[\kappa_1,\kappa_2]$ for some constants $\kappa_1,\kappa_2>0$.

 A rate-distortion codebook with rate $R$ and block length $n$  is a set of $e^{nR}$ length-$n$  codewords, denoted $\{\bfsh(1),\ldots, \bfsh(e^{nR}) \}$.  The quality of reconstruction is measured through an average squared-error distortion criterion
\[ d_n(\bfs, \bfsh)=\abs{\bfs -\bfsh}^2= \frac{1}{n}\sum_{i=1}^n(S_i - \hat{S}_i)^2 , \]
where $\bfsh$ is the codeword chosen to represent the source sequence $\bfs$. For this distortion criterion, an optimal encoder maps each source sequence to the codeword nearest to it in Euclidean distance.

For an i.i.d Gaussian source distributed as $\mc{N}(0,\sigma^2)$, the  rate-distortion function $R^*(D)$  -- the minimum rate  for which the distortion can be bounded by $D$ with high-probability -- is given by  \cite{CoverThomas}
\be
\begin{split}
R^*(D) & = \min_{p_{\hat{S}|S}: E(S-\hat{S})^2 \leq D} I(S;\hat{S}) \\
& = \left\{ \begin{array}{l c} \frac{1}{2} \log \frac{\sigma^2}{D} \ \text{nats/sample} & \text{  for $D < \sigma^2$}, \\
0 & \text{  for $D \geq \sigma^2$}.
\end{array}
\right.
\end{split}
\label{eq:gaussian_rd}
\ee
This rate can be achieved through Shannon-style random codebook selection: pick each codeword independently as an i.i.d  Gaussian random vector distributed as $\mc{N}(0, \sigma^2-D)$. Both the storage and encoding complexities of such a codebook grow exponentially with block length. Lattice-based codes for lossy compression have been widely studied (e.g \cite{EyForney93,Zamir02}), and have a compact representation, i.e., low storage complexity. There are computationally  efficient quantizers for certain classes of lattice codes, but  the high-dimensional lattices needed to approach the rate-distortion bound have exponential encoding complexity \cite{Zamir02}. We also note that for sources with finite alphabet, various coding techniques have been proposed recently to approach the rate-distortion bound with computationally feasible encoding and decoding \cite{GuptaVerduWeiss,KontGioran,JalaliWeiss, GuptaVerdu09, WainManeva10,polarrd}.

Sparse regression codes for lossy compression were first considered in \cite{KontSPARC} where some preliminary results were presented. In this paper, we analyze the performance of these codes under optimal (minimum-distance) encoding. The main contributions are the following.
\begin{itemize}
\item[-] We obtain an achievable SPARC rate-distortion function and error exponent for compression of ergodic sources (with known first and second moments) under the squared-error distortion criterion.

\item[-] For the special case of compressing i.i.d Gaussian sources, our results imply that SPARCs achieve the optimal squared-error distortion with the optimal error exponent for all rates above a specified value (approximately $0.797$ nats or $1.15$ bits per sample).\footnote{Below this rate, the achieved distortion is higher than the optimal value; however, we think that this is due  to a limitation of the proof technique. We conjecture that SPARCs  achieve the optimal rate-distortion trade-off for all rates with minimum-distance encoding.}

\item[-] Our results show that SPARCs (at least for rates greater than $1.15$ bits) are essentially as good as random i.i.d Gaussian codebooks  in terms of distortion-rate function, error exponents, and robustness. By robustness, we mean that a SPARC designed to compress an i.i.d Gaussian source with variance $\sigma^2$ to distortion $D$  can compress any ergodic source with variance less than or equal to $\sigma^2$ to distortion $D$. This property is also satisfied by random i.i.d Gaussian codebooks \cite{Lapidoth97,SakMismatch1,SakMismatch2}.\footnote{In fact, Lapidoth \cite{Lapidoth97} also shows that for any ergodic source of a given variance,  with a Gaussian random codebook one cannot attain a mean-squared distortion smaller than the distortion-rate function of an i.i.d Gaussian source with the same variance.}
These results show that the sparse regression ensemble has  good covering properties, with the advantage of much smaller codebook storage complexity than the i.i.d random ensemble (polynomial vs. exponential  in block-length).
\end{itemize}

A consequence of the SPARC's compact representation is that its codewords are dependent.  We deal with the dependence using two techniques from the literature on random graphs. For the rate-distortion analysis,  the well-known second-moment method \cite{JansonBook} suffices. The error exponent requires a more refined analysis and we use Suen's inequality \cite{janson98,JansonBook}, an exponential bound on the tail probability of a sum of dependent indicator random variables. This technique may be of independent interest and useful in other problems in information theory.

In Section \ref{sec:sparc}, we describe the sparse regression codebook along with the encoding and decoding procedure. The main results, describing the rate-distortion and error exponent performance of SPARCs, are stated in Section \ref{sec:main_results}. The proofs of these results are given in the next two sections. We first derive the rate-distortion function  in Section \ref{sec:proof_rd_thm} using the second moment method, which highlights the features that make SPARCs  more challenging to analyze than i.i.d random codebooks.
The analysis is then refined using Suen's inequality in Section \ref{sec:proof_err_exp} to obtain the error exponent. Section \ref{sec:conc} concludes the paper with a brief discussion of why the proof techniques do not yield the optimal distortion-rate function for rates smaller than $1.15$ bits.

\section{Sparse Regression Codes} \label{sec:sparc}
\begin{figure}[t]
\centering
\includegraphics[width=3.4in]{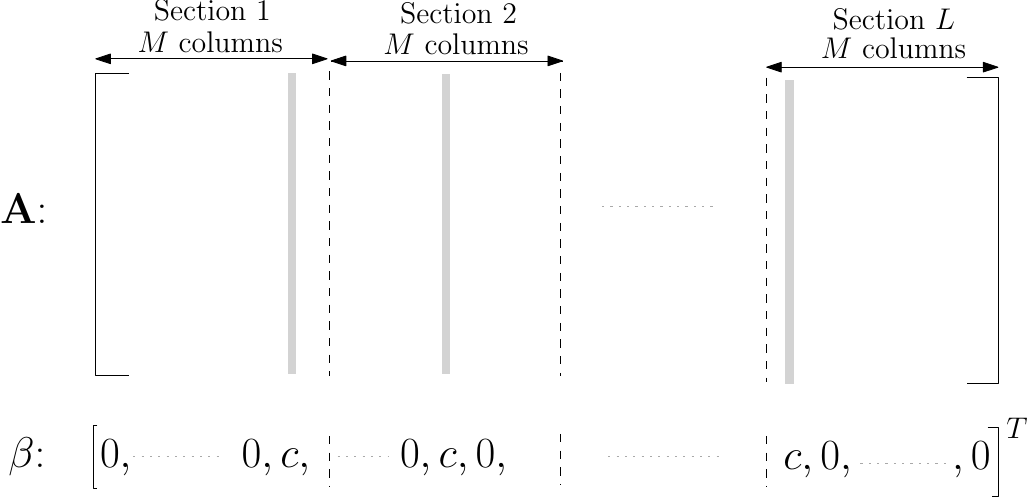}
\caption{\small{$\mathbf{A}$ is an $n \times ML$ matrix and $\beta$ is a $ML \times 1$ binary vector. The positions of the non-zeros in $\beta$  correspond to the gray columns of $\mathbf{A}$ which combine to form the codeword $\mathbf{A}\beta$.}}
\vspace{-5pt}
\label{fig:sparserd}
\end{figure}

A sparse regression code (SPARC) is defined in terms of a design matrix $\mathbf{A}$ of dimension $n \times ML$ whose entries are i.i.d. $\mathcal{N}(0,1)$.
Here $n$ is the block length and $M$ and $L$ are integers whose values will be specified shortly in terms of $n$ and the rate $R$.  As shown in Figure \ref{fig:sparserd}, one can think of the matrix $\mathbf{A}$ as composed of $L$ sections with $M$ columns each. Each codeword is a linear combination of $L$ columns, with one column from each section.
Formally, a codeword can be expressed as  $\mathbf{A} \beta$, where $\beta$ is  a binary-valued $ML \times 1$ vector $(\beta_1, \ldots, \beta_{ML})$ with the following property:  there is exactly one non-zero $\beta_i$ for  $1 \leq i \leq M$, one non-zero $\beta_i$ for $M+1 \leq i \leq 2M$, and so forth.  The non-zero values of $\beta$ are all set equal to $c \triangleq \frac{\gamma}{\sqrt{L}}$ where $\gamma$ is a constant that will be specified later (at the beginning of Section \ref{sec:proof_rd_thm}).  Denote the set of all $\beta$'s that satisfy this property by $\mcb$.

\emph{Minimum-distance Encoder}: This is defined by a mapping $g: \mathbb{R}^n \to \mcb$. Given the source sequence $\bfs$, the encoder determines the $\beta$ that produces the codeword closest in Euclidean distance, i.e.,
 \[ g(\bfs) = \underset{\beta \in \mcb}{\operatorname{argmin}} \ \norm{\bfs - \mathbf{A}\beta}.\]

\emph{Decoder}: This is a mapping $h: \mcb \to \mathbb{R}^n$. On receiving ${\beta} \in \mcb$ from the encoder, the decoder produces reconstruction $h(\beta) = \mathbf{A}\beta$.

Since there are $M$ columns in each of the $L$ sections, the total number of codewords is $M^L$. To obtain a compression rate of $R$ nats/sample, we therefore need
\be
M^L = e^{nR}.
\label{eq:ml_nR}
\ee
There are several choices for the pair $(M,L)$ which satisfy this. For example, $L=1$ and $M=e^{nR}$ recovers the Shannon-style random codebook in which the number of columns in the dictionary $\mathbf{A}$ is $e^{nR}$, i.e., exponential in $n$. For our constructions, we choose $M=L^b$ for some $b>1$ so that \eqref{eq:ml_nR} implies
\be   L \log L = {nR}/{b}. \label{eq:rel_nL} \ee
 Thus $L$ is now $\Theta\left(\tfrac{n}{\log n}\right)$, and the number of columns $ML$ in the dictionary $\mathbf{A}$ is now $ \Theta\left(\left(\tfrac{n}{\log n}\right)^{b+1}\right)$, a \emph{polynomial}  in $n$. This reduction in storage complexity can be harnessed to develop computationally efficient encoders for the sparse regression code. This is discussed in \cite{RVGaussianFeasible}.

The code structure automatically yields low decoding complexity. The encoder can represent the chosen $\beta$ with $L$ binary sequences of $\log_2 M$ bits each. The $i$th binary sequence indicates the position of the non-zero element in section $i$. Thus the decoder complexity involved in locating the $L$ non-zero elements using the received bits is $L \log_2 M$. Reconstructing the codeword then involves $L$ additions per source sample.

Since each codeword in a SPARC is a linear combination of $L$ columns of $\mathbf{A}$ (one from each section),  codewords sharing one or more common columns in the sum will be dependent.  Also, SPARCs are not  linear codes since the sum of two codewords does not equal another codeword in general.

\section{Main Results} \label{sec:main_results}
In this section, we discuss the rate-distortion performance and error exponent of SPARCs under minimum-distance encoding.

\subsection{Rate-Distortion Performance of SPARC}

 The probability of error at distortion-level $D$ of a rate-distortion code $\mathcal{C}_n$ with block length $n$ and encoder and decoder mappings $g,h$  is
\be P_{e}(\mathcal{C}_n, D) = P\left(\abs{\bfs - h(g(\bfs))}^2 > D\right).  \ee
\begin{defi}
A rate $R$ is achievable at distortion level $D$  if there exists a sequence of SPARCs $\{\mathcal{C}_n\}_{n=1,2,\ldots}$ such that
$\lim_{n \to \infty} P_{e}(\mathcal{C}_n, D) =0$
where for all $n$, $\mathcal{C}_n$ is a rate $R$ code defined by an $n \times L_n M_n$ design matrix whose parameter $L_n$ satisfies \eqref{eq:rel_nL} with a fixed $b$ and $M_n=L_n^b$.
\end{defi}

The rate-distortion performance of SPARCs is given by the following theorem.
\begin{thm}
Let $\bfs$ be a drawn from any ergodic source with mean $0$ and variance $\sigma^2$. For $D \in (0, \sigma^2)$, let $R_{sp}(D) = \max\{ \frac{1}{2} \log \frac{\sigma^2}{D}, \ 1-\frac{D}{\sigma^2} \}$. Then for all rates $R > R_{sp}(D)$ and
\[ b>\frac{2.5 R}{R- 1  + D/\sigma^2},  \] there exists a sequence of rate $R$  SPARCs $\{\mc{C}_n\}_{n=1,2 \ldots}$  for which $\lim_{n \to \infty} P_{e}(\mathcal{C}_n, D) =0$, where $\mathcal{C}_n$ is defined by an $n  \times L_n M_n$ design matrix with $L_n$ determined by \eqref{eq:rel_nL} and $M_n =L_n^b$.
\label{thm:rd_sparc}
\end{thm}
\begin{IEEEproof} In Section \ref{sec:proof_rd_thm}. \end{IEEEproof}

\emph{Remarks}:
\begin{enumerate}
\item Theorem \ref{thm:rd_sparc} implies that SPARCs achieve the optimal rate-distortion function $R^*(D) = \frac{1}{2} \log \frac{\sigma^2}{D}$ for $0 < \frac{D}{\sigma^2}  < x^*$ where $x^* \approx 0.2032$ is the solution of the equation $1-x + \frac{1}{2} \log x =0$.  Equivalently, SPARCs with rate at least $1-x^* \approx 0.797$ nats ($1.15$ bits) achieve the optimal distortion-rate function with minimum-distance encoding.

\item For $x^*\leq \frac{D}{\sigma^2} \leq 1$, the minimum achievable rate of Theorem \ref{thm:rd_sparc} $(1-\frac{D}{\sigma^2})$  is larger than the optimal rate-distortion function. In this region, $R_{sp}(D)$ can also be achieved by time-sharing between the points $\frac{D}{\sigma^2}=x^*$ and $\frac{D}{\sigma^2}=1$.

\item The proofs of the results (Lemmas \ref{lem:EU1_asymp} and \ref{lem:Tz2} in particular) show that a SPARC designed to compress a source of variance $\sigma^2$ with distortion $D$ yields distortion less than or equal to $D$ with high probability on any source sequence whose empirical variance is less than or equal to $\sigma^2$.
\end{enumerate}

The rate-distortion performance described by Theorem \ref{thm:rd_sparc} is shown in Figure \ref{fig:sparc_perf}.

\begin{figure}[t]
\centering
\includegraphics[width=3.4in]{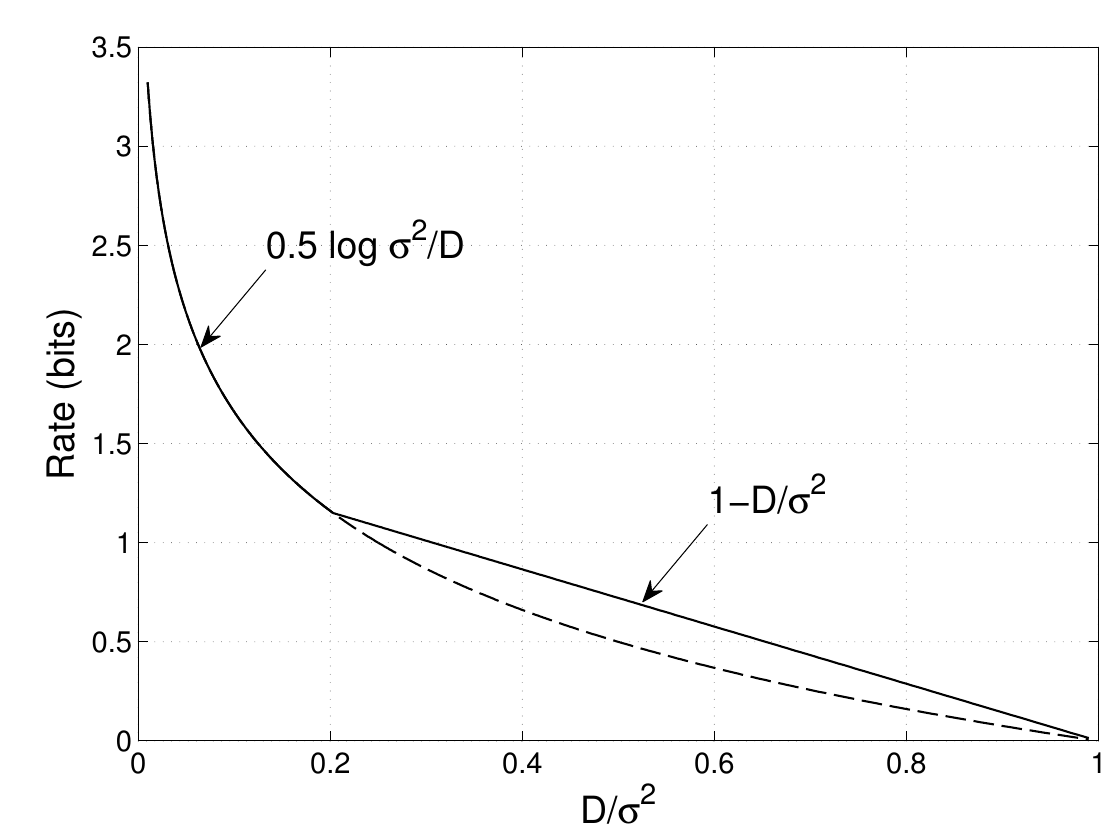}
\caption{\small The achievable rate $R_{sp}(D)$ of Theorem $1$. The dashed line shows the Shannon rate-distortion function $\frac{1}{2}\log\frac{\sigma^2}{D}$ which coincides with $R_{sp}(D)$ for $D/\sigma^2 < x^* \approx 0.203$.}
\label{fig:sparc_perf}
\end{figure}

\subsection{Error Exponent of SPARC}
We begin with some background on error exponents.
\begin{defi}
The error exponent at distortion-level $D$ of a sequence of rate $R$ codes $\{\mathcal{C}_n\}_{n=1,2,\ldots}$ is given by
\be r(D,R) = - \limsup_{n \to \infty} \frac{1}{n} \log  P_{e}(\mathcal{C}_n,D). \ee
The optimal error exponent for a rate-distortion pair $(R,D)$ is the supremum of the error exponents over all sequences of codes with rate  $R$, at distortion-level $D$.
\end{defi}

 For a given rate $R$ and distortion-level $D$, the error exponent describes the asymptotic behavior of the probability of error with growing block length; bounds on the probability of error for finite block lengths were obtained in \cite{SakrisonFin,IngberKochman,KostinaV12}. The optimal error exponent was obtained by Marton \cite{MartonRD74} for discrete memoryless sources and was extended to Gaussian 
 sources by Ihara and Kubo \cite{IharaKubo00}. 

\begin{fact} \cite{IharaKubo00}
\label{fact:ihara}
For an i.i.d Gaussian source distributed as $\mathcal{N}(0, \sigma^2)$ and  squared-error distortion criterion, the optimal error exponent at rate $R$ and distortion-level $D$ is
\be
r^*(D,R) = \left\{
\begin{array}{ll}
 \frac{1}{2} \left( \frac{y^2}{\sigma^2} - 1 - \log \frac{y^2}{\sigma^2} \right) & \quad R> R^*(D) \\
 0 & \quad R \leq R^*(D)
\end{array}
\right.
\label{eq:opt_exp}
\ee
where  \be y^2 = D e^{2R}. \label{eq:a2R} \ee
\end{fact}

For $R> R^*(D)$, the exponent in \eqref{eq:opt_exp} is the Kullback-Leibler divergence between two zero-mean Gaussian distributions, the first with variance $y^2$ and the second with variance $\sigma^2$. It is shown in \cite{IharaKubo00}  that at rate $R$, we can compress all sequences which have empirical variance less than $y^2$ to within distortion $D$ with \emph{double-exponentially} decaying probability of error. Consequently, the dominant error event is  obtaining a source sequence with empirical variance greater than $y^2$, which has exponent given by \eqref{eq:opt_exp}.

We now characterize the error exponent performance of SPARCs. As in Theorem \ref{thm:rd_sparc}, let
\be
R_{sp}(D) = \max\left\{ \frac{1}{2} \log \frac{\sigma^2}{D}, \ 1-\frac{D}{\sigma^2} \right\}.
\label{eq:Rsp_def}
\ee
\begin{thm}
Let $\bfs$ be drawn from an ergodic source with mean zero and variance $\sigma^2$,  and let $D \in (0, \sigma^2)$. For any rate $R >0$, define
\be
a^2 = \left\{
\begin{array}{ll}
D e^{2R} & \quad R \geq 1-x^* \\
\frac{D}{1-R} & \quad 0 < R < 1-x^*
\end{array}
\right.
\label{eq:a2def}
\ee
where $x^* \approx 0.2032$ is the solution of the equation $1-x + \frac{1}{2} \log x =0$.
Fix any $\e  \in (0, a^2)$ and let
\be b>\frac{3.5 R}{R- 1  + D/(a^2-\e)}. \label{eq:bmin_exp} \ee  Then there exists a sequence of rate $R$ SPARCs $\{\mc{C}_n\}_{n=1,2 \ldots}$ whose probability of error at distortion-level $D$ can be bounded as follows for all sufficiently large $n$. 
\be
P_e(\mc{C}_n, D)  \leq P(\abs{\bfs}^2 \geq a^2 - \e) + \exp\left( - \Theta(n^t) \right),
\ee
where $t$ is a constant strictly greater than $1$.

For each $n$, $\mathcal{C}_n$ is defined by an $n  \times L_n M_n$ design matrix with $L_n$ determined by \eqref{eq:rel_nL} and $M_n =L_n^b$.
\label{thm:err_exp}
\end{thm}
\begin{IEEEproof} In Section \ref{sec:proof_err_exp}. \end{IEEEproof}

\emph{Remark}: From  \eqref{eq:a2def}, $R$ can be expressed in terms of $D$ and $a^2$ as
\be
R = \max \left\{ \frac{1}{2} \log \frac{a^2}{D}, \ 1-\frac{D}{a^2}  \right\}.
\label{eq:R_a2}
\ee

\begin{corr}
Let $\bfs$ be drawn from an i.i.d Gaussian source with mean zero and variance $\sigma^2$. For any $\e \in (0,a^2)$, there exists a sequence of SPARCs with design matrix parameter $b$ satisfying \eqref{eq:bmin_exp} that achieves the error exponent
\be
r_{sp}(D,R) = \left\{
\begin{array}{ll}
 \frac{1}{2} \left( \frac{a^2 - \e}{\sigma^2} - 1 - \log \frac{a^2-\e}{\sigma^2} \right) & \quad R> R_{sp}(D) \\
 0 & \quad R \leq R_{sp}(D)
\end{array}
\right.
\label{eq:rsp_exp}
\ee
where $a^2$ is given by \eqref{eq:a2def}.

In particular, SPARCs attain the optimal error exponent for all rates greater than $1-x^* \approx 0.797$ nats ($1.15$ bits).
\label{corr:err_exp}
\end{corr}
\begin{IEEEproof}
From Theorem \ref{thm:err_exp}, we know that for any $\e \in (0,a^2)$, there exists a sequence of rate $R$ SPARCs $\{ \mc{C}_n\}$ for which
\be
P_e(\mc{C}_n, D)  \leq P(\abs{\bfs}^2 \geq a^2 - \e) \left(1 + \frac{\exp( - \Theta(n^t))}{P(\abs{\bfs}^2 \geq a^2-\e)}\right)
\label{eq:pe_rew}
\ee
for sufficiently large $n$, where $t >1$. For $\bfs$ that is  i.i.d $\mc{N}(0, \sigma^2)$, Cram{\'e}r's large deviation theorem \cite{Den2008LD,DemboZbook} yields
\be \lim_{n \to \infty} - \frac{1}{n} \log P(\abs{\bfs}^2 \geq y^2) =  \frac{1}{2} \left( \frac{y^2}{\sigma^2} -1 - \log \frac{y^2}{\sigma^2} \right) \label{eq:gaussian_ld} \ee
for $y^2 > \sigma^2$. Thus in \eqref{eq:pe_rew},  $P(\abs{\bfs}^2 \geq a^2-\e)$ thus decays exponentially with $n$ for $(a^2-\e) > \sigma^2$, while $\exp( - \Theta(n^t))$ decays \emph{faster} than exponentially since $t>1$. Therefore,  for $(a^2-\e) > \sigma^2$ we have
\be
\begin{split}
& \liminf_{ n \to \infty} \, \frac{-1}{n} \log P_e(\mc{C}_n, D) \\
&    \geq  \liminf_{ n \to \infty}  \frac{-1}{n}\Bigg[ \log P(\abs{\bfs}^2 \geq a^2-\e) \\
& \qquad \qquad   \qquad + \log \Big(1 + \frac{\exp( - \Theta(n^t))}{P(\abs{\bfs}^2 \geq a^2-\e)}\Big) \Bigg] \\
& = \frac{1}{2} \left( \frac{a^2 - \e}{\sigma^2} -1 - \log \frac{a^2-\e}{\sigma^2} \right).
\end{split}
\ee
Since $\e >0$ is arbitrary, we conclude that the error exponent in \eqref{eq:rsp_exp} can be achieved.
\end{IEEEproof}

From \eqref{eq:bmin_exp}, we note that  larger values of the design parameter $b$ are required to achieve  error exponents closer to the optimal value (i.e., smaller values of $\e$ in Corollary \ref{corr:err_exp}).

\begin{figure*}[b]
\vspace*{0pt} \hrulefill
\normalsize
\setcounter{mytempeqncnt}{\value{equation}} 
\setcounter{equation}{25}
\begin{equation}
\begin{split}
&P( \mc{E}(\bfs) \mid \abs{\bfs}^2 = z^2 ) =  1 - P\left( \sum_{i} U_i > 0 \mid \abs{\bfs}^2 = z^2 \right)\\
&< \ \frac{\expec[U_1 \mid \abs{\bfs}^2 = z^2] \ + \
\sum_{r=1}^{L-1} {L \choose r} (M-1)^{L-r} \expec[U_1 U_j \mid \mathcal{F}_{1j}(r), \abs{\bfs}^2 = z^2 ]}{\expec[U_1 \mid \abs{\bfs}^2 = z^2]   +
(M-1)^L (\expec[U_1 \mid \abs{\bfs}^2 = z^2])^2 \ + \
\sum_{r=1}^{L-1} {L \choose r} (M-1)^{L-r} \expec[U_1 U_j \mid \mathcal{F}_{1j}(r), \abs{\bfs}^2 = z^2 ] }.
\end{split}
\label{eq:pe_long_bound1}
\end{equation}
\setcounter{equation}{\value{mytempeqncnt}}
\end{figure*}

\section{Proof of Theorem \ref{thm:rd_sparc}} \label{sec:proof_rd_thm}
Fix a rate $R > R_{sp}(D)$, and $b$ greater than the minimum value specified by the theorem.  Let $a^2$ be defined by \eqref{eq:a2def}. Since  $R > R_{sp}(D)$, it follows by comparing the expressions in \eqref{eq:Rsp_def} and \eqref{eq:R_a2} that  $a^2 > \sigma^2$.
Let $\rho^2$ be any number such that $\sigma^2 < \rho^2 < a^2$.

\emph{Code Construction}:
 For each block length $n$, pick $L$ as specified by \eqref{eq:rel_nL} and $M=L^b$.
Construct an $n \times ML$ design matrix $\mathbf{A}$ with entries drawn i.i.d $\mathcal{N}(0, 1)$. The codebook consists of all vectors $\mathbf{A} \beta$ such that $\beta \in \mcb$ and the non-zero entries of $\beta$ are all equal to $\frac{\gamma}{\sqrt{L}}$ where $\gamma^2 = \rho^2 -D$.

\emph{Encoding and Decoding}: If the source sequence $\bfs$ is such that $\abs{\bfs}^2 \geq \rho^2$, then the encoder declares an error. Else, it finds
\[ \hat{\beta} \triangleq g(\bfs) = \underset{\beta \in \mcb}{\operatorname{argmin}} \ \norm{\bfs - \mathbf{A}\beta}^2. \]
The decoder receives $\hat{\beta}$ and reconstructs $\bfsh = \mathbf{A} \hat{\beta}$.

\emph{Error Analysis}:
Denoting the probability of error for this random code by $P_{e,n}$, we have
\be
\begin{split}
P_{e,n}& \leq   1 \cdot  P(\abs{\bfs}^2 \geq \rho^2) +   \int_{0}^{\rho^2} P(\mc{E}(\bfs) \mid \abs{\bfs}^2 = z^2 ) d\nu(z^2) \\
&  \leq P(\abs{\bfs}^2 \geq \rho^2) +  \max_{z^2 \in (D, \rho^2)} P(\mc{E}(\bfs) \mid \abs{\bfs}^2 = z^2 ).
\end{split}
\label{eq:err_bound}
\ee
where $\mc{E}(\bfs)$ is the event that the minimum of $\abs{\bfs - \mathbf{A}\beta}^2$ over $\beta \in \mcb$ is greater than $D$, and   $\nu(\abs{\bfs}^2)$ is the  distribution of the random variable $\abs{\bfs}^2$.  The maximum in the second term  can be restricted to $z^2 \in (D, \rho^2)$ since source sequences $\bfs$ with empirical second moment less than $D$ can be trivially compressed using the all-zero codeword. The addition of this extra codeword to the codebook affects the rate in a negligible way.

Since $\rho^2> \sigma^2$, the ergodicity of the source guarantees that
\be
\lim_{n \to \infty} P(\abs{\bfs}^2 \geq \rho^2) = 0.
\label{eq:s_erg}
\ee
The remainder of the proof is devoted to bounding the second term in \eqref{eq:err_bound}. Recall that
 \be
 \begin{split}
 & P\left(\mc{E}(\bfs) \mid \abs{\bfs}^2 = z^2 \right) \\
 &= P(\abs{\bfsh(i) - \bfs}^2 >  D, \ i=1,\ldots,e^{nR} \mid \abs{\bfs}^2= z^2 )
 \end{split}
 \label{eq:peps}
 \ee
 where $\bfsh(i)$ is the $i$th codeword in the sparse regression codebook. Define indicator random variables
 $U_i(\bfs)$ for $i=1,\ldots, e^{nR}$ as follows:
 \be
U_i(\bfs) = \left\{
\begin{array}{ll}
1 & \text{ if } \abs{\bfsh(i) - \bfs}^2 \leq D,\\
0 & \text{ otherwise}.
\end{array} \right.
\label{eq:ui_def}
\ee
From \eqref{eq:peps}, it is seen that
\be
P( \mc{E}(\bfs) \, |  \, \abs{\bfs}^2 = z^2 ) = P\Big(\sum_{i=1}^{e^{nR}} U_i(\bfs) =0 \, | \, \abs{\bfs}^2 = z^2 \Big).
\label{eq:sum_ui}
\ee

For a fixed $\bfs$,  the $U_i(\bfs)$'s are dependent. Suppose that the codewords $\bfsh(i),\bfsh(j)$ respectively correspond to the binary vectors $\hat{\beta}(i), \hat{\beta}(j) \in \mcb$. Recall that each vector in $\mcb$ is uniquely defined by the position of the non-zero value in each of the $L$ sections.   If $\hat{\beta}(i)$ and $\hat{\beta}(j)$ overlap in $r$ of their non-zero positions, then the column sums forming codewords $\bfsh(i)$ and $\bfsh(j)$ will share $r$ common terms.
For each codeword $\bfsh(i)$, there are ${L \choose r} (M-1)^{L-r}$ other codewords which share exactly $r$ common terms with $\bfsh(i)$, for $0\leq r \leq L-1$. In particular, there are $(M-1)^L$ codewords that are pairwise independent of $\bfsh(i)$.

We now obtain an upper bound for the probability in \eqref{eq:sum_ui} using the second moment method \cite{JansonBook}. For any non-negative random variable $X$, the second moment method lower bounds the probability of the event $X > 0$ as
\be
P(X > 0) \geq {(\expec X)^2}{\big{/}\,}{\expec [ X^2]}.
\label{eq:2nd_mom_method}
\ee
\eqref{eq:2nd_mom_method} follows from the Cauchy-Schwarz inequality $ (\expec[X Y])^2 \leq \expec[X^2] \expec[Y^2]$ by substituting $Y=\mathbf{1}_{\{X>0\}}$.
Applying \eqref{eq:2nd_mom_method} with $X=\sum_{i=1}^{e^{nR}} U_i(\bfs)$,  we have
\be
\begin{split}
& P\Big( \sum_{i=1}^{e^{nR}} U_i(\bfs)> 0 \mid \abs{\bfs}^2 = z^2 \Big) \\
 & \geq {\Big(\expec[\sum_{i} U_i(\bfs) | \, \abs{\bfs}^2 = z^2]\Big)^2} \Big{/} \,
  {\expec\Big[\big(\sum_{i} U_i(\bfs)\big)^2 | \, \abs{\bfs}^2 = z^2\Big]} \\
&{=} \, {e^{2nR} \left(\expec[ U_1(\bfs) \mid \abs{\bfs}^2 = z^2]\right)^2} \Big{/}
\Big( e^{nR} \expec\left[U_1(\bfs) \mid \abs{\bfs}^2 = z^2 \right] \\
& \hspace{1.05in} + \ e^{nR} \sum_{j \neq 1}\expec\left[U_1(\bfs) U_j(\bfs) \mid \abs{\bfs}^2 = z^2 \right] \Big)
\end{split}
\label{eq:2nd_mom_chain1}
\ee
where the equality  is due to the symmetry of the code construction. For brevity, from here on we suppress the dependence of the $U$'s on $\bfs$.

Using $\mathcal{F}_{ij}(r)$ to denote the event that $\bfsh(i),\bfsh(j)$ share $r$ common terms, \eqref{eq:2nd_mom_chain1} can be written as
\be
P\left( \sum_{i} U_i > 0 \mid \abs{\bfs}^2 = z^2 \right) \geq \frac{e^{nR} \left(\expec[ U_1 \mid \abs{\bfs}^2 = z^2]\right)^2}{\text{Denom}}
\label{eq:2nd_mom_chain2}
\ee
where the denominator is given by
\be
\begin{split}
& \text{Denom}  = \expec[U_1 \mid \abs{\bfs}^2 = z^2] \\
& \qquad +  \sum_{r=0}^{L-1} {L \choose r} (M-1)^{L-r} \expec[U_1 U_j \mid \mathcal{F}_{1j}(r), \abs{\bfs}^2 = z^2] \\
& =  \expec[U_1 \mid \abs{\bfs}^2 = z^2] \, + \, (M-1)^L (\expec[U_1 \mid \abs{\bfs}^2 = z^2])^2  \\
& \quad + \sum_{r=1}^{L-1} {L \choose r} (M-1)^{L-r} \expec[U_1 U_j \mid \mathcal{F}_{1j}(r), \abs{\bfs}^2 = z^2].
\end{split}
\ee
Using \eqref{eq:2nd_mom_chain2}, the probability in \eqref{eq:sum_ui} can be bounded as shown in \eqref{eq:pe_long_bound1} at the bottom of this page.
\addtocounter{equation}{1}
Dividing the numerator and denominator of the RHS of \eqref{eq:pe_long_bound1} by \[(M-1)^L (\expec[U_1 \mid \abs{\bfs}^2 = z^2])^2, \] we obtain
\be
\begin{split}
&P( \mc{E}(\bfs) \mid \abs{\bfs}^2 = z^2 ) \\
&\stackrel{(a)}{<} \  \frac{\left( (M-1)^L \expec[ U_1 \mid \abs{\bfs}^2 = z^2] \right)^{-1} + T(z^2)}
{1 + \left( (M-1)^L \expec[U_1 \mid \abs{\bfs}^2 = z^2 ] \right)^{-1} + T(z^2) }
\end{split}
\label{eq:2nd_mom_chain3}
\ee
where 
\be
T(z^2) \triangleq  \frac{\sum_{r=1}^{L-1} {L \choose r} (M-1)^{L-r} \expec\left[U_1 U_j | \mathcal{F}_{1j}(r), \abs{\bfs}^2 = z^2 \right]}
{(M-1)^L \left(\expec[U_1 \mid \abs{\bfs}^2 = z^2]\right)^2}.
\label{eq:tz2_def}
\ee
\eqref{eq:2nd_mom_chain3} implies that  $P(\mc{E}(\bfs) \mid \abs{\bfs}^2 = z^2 )$ will go to $0$ as $n \to \infty$ if
\begin{enumerate}
\item $(M-1)^L \expec[ U_1 \mid \abs{\bfs}^2 = z^2]  \to \infty$ as $n \to \infty$ and

\item $T(z^2) \to 0$ as  $n \to \infty$.
\end{enumerate}
We now derive the conditions under which both of the above hold.
\begin{lem}
(a) $(M-1)^L \, \expec [ U_1 \, | \, \abs{\bfs}^2 = z^2] \  \stackrel{n \to \infty}{\longrightarrow} \ \infty$  if
\be
R > f(z^2) \triangleq \frac{D+z^2}{2\gamma^2} - \frac{Dz^2}{A(z^2) \gamma^2} - \frac{A(z^2)}{4\gamma^2} -\frac{1}{2} \ln\frac{A(z^2)}{2z^2}
\label{eq:fz2_cond}
\ee
where $\gamma^2= \rho^2-D $ and $A(z^2) \triangleq \sqrt{(\gamma^4 + 4z^2D)} - \gamma^2$.

(b) $(M-1)^L \expec [ U_1  | \, \abs{\bfs}^2 = z^2] \  \stackrel{n \to \infty}{\longrightarrow} \ \infty$ for all $z^2 \in (D, \rho^2)$  since
$R> \frac{1}{2} \log \frac{\rho^2}{D}$.  
\label{lem:EU1_asymp}
\end{lem}
\begin{IEEEproof}
The first part is an application of Cram{\'e}r's large deviation theorem \cite{DemboZbook}. We have
\be
\begin{split}
& \expec [ U_1 \mid \abs{\bfs}^2 = z^2]  = P(\abs{\bfsh(1) - \bfs}^2 \leq  D \mid \abs{\bfs}^2 = z^2) \\
&  = P\left(\frac{1}{n}\sum_{k=1}^n (\hat{S}_k(1)- S_k)^2 \leq  D  \mid  \abs{\bfs}^2 = z^2\right) \\
& = P\left(\frac{1}{n}\sum_{k=1}^n (\hat{S}_k(1) - z)^2 \leq  D\right)
\end{split}
\label{eq:EU1_simp}
\ee
where the last equality is due to the rotational invariance of the distribution of $\bfsh(1)$, i.e.,  $\bfsh(1)$ has the same joint distribution as $\mathbf{O}\bfsh(1)$ for any orthogonal (rotation) matrix $\mathbf{O}$. The large deviation rate-function for the right-side of \eqref{eq:EU1_simp} can be computed from Cram{\'e}r's theorem to obtain
\be
\begin{split}
& \lim_{n \to \infty} -\frac{1}{n} \log \expec [ U_1 \mid \abs{\bfs}^2 = z^2] \\
& =  \frac{D+z^2}{2\gamma^2} - \frac{Dz^2}{A(z^2)\gamma^2} - \frac{A(z^2)}{4\gamma^2} -\frac{1}{2} \ln\frac{A(z^2)}{2z^2} = f(z^2)
\label{eq:EU1_ld}
\end{split}
\ee
where $A(z^2) \triangleq \sqrt{(\gamma^4 + 4z^2D)} - \gamma^2$. Since $M^L = e^{nR}$, \eqref{eq:EU1_ld} implies that 
\be M^L \, \expec [ U_1 \, | \, \abs{\bfs}^2 = z^2] \  \stackrel{n \to \infty}{\rightarrow} \ \infty \ee
if $R > f(z^2)$. Part (a) of the lemma follows by noting that
\be
\frac{(M-1)^L}{M^L} = (1 - L^{-b})^L \ \stackrel{n \to \infty}{\rightarrow} \ 1
\ee because $b>1$. 

For part (b),  we show that $f(z^2)$ (the RHS of \eqref{eq:fz2_cond}) is increasing in $z^2$. Indeed its derivative with respect to $z^2$ is equal to
\[ \frac{1}{2\gamma^2} \left(1 -\frac{2D}{\sqrt{\gamma^4 + 4 z^2 D}}\right) +  \frac{A(z^2)}{4z^2 \sqrt{\gamma^4 + 4 z^2 D}} ,\]
which is strictly positive for $z^2 \geq D$.  Therefore the maximum value of the RHS in the interval $(D, \rho^2)$ is attained at $z^2 =\rho^2$ where it is equal to $\frac{1}{2} \log \frac{\rho^2}{D}$.
\end{IEEEproof}
\begin{lem}
$T(z^2) \ \stackrel{n \to \infty}{\longrightarrow} \ 0$  for all $z^2 \in (D, \rho^2)$ if
\be
b > \frac{2.5 R}{R - \left(1 - D/\rho^2 \right)}.
\label{eq:brd_cond}
\ee
\label{lem:Tz2}
\end{lem}
\begin{IEEEproof} In Appendix \ref{app:lemma}. \end{IEEEproof}

Using Lemmas \ref{lem:EU1_asymp} and \ref{lem:Tz2} in \eqref{eq:2nd_mom_chain3}, we see that for all $z^2 \in (D, \rho^2)$,
$P(\mc{E}(\bfs) \mid \abs{\bfs}^2 = z^2 ) \stackrel{n \to \infty}{\longrightarrow} 0$  provided $b$ satisfies the condition in \eqref{eq:brd_cond}. Together with \eqref{eq:s_erg}, this implies that $P_{e,n}$ in \eqref{eq:err_bound} goes to $0$ as $n \to \infty$. Since $\rho^2$ is an any number in the interval
$(\sigma^2, a^2)$, we can choose $\rho^2$  to be arbitrarily close to $\sigma^2$.   Thus we have shown that $P_{e,n} \to 0$ for any  rate $R > R_{sp}(D)$ as long as
\[  b > \frac{2.5 R}{R - \left(1 - D/\sigma^2 \right)}. \]

\section{Proof of Theorem \ref{thm:err_exp}} \label{sec:proof_err_exp}
Fix a rate $R > R_{sp}(D)$, and choose a value of $b$ greater than the minimum  specified by the theorem.  Let $a^2$ be defined by \eqref{eq:a2def}, and
$\rho^2$ be any number such that $\sigma^2 < \rho^2 < a^2$.

 With the code construction, encoding and decoding as described in Section \ref{sec:proof_rd_thm}, the error exponent is obtained through a more refined analysis of the probability of error of the random sparse regression codebook.
We will use Suen's correlation inequality to obtain a bound on the tail probability of the second term in \eqref{eq:err_bound}. This bound is sharper than the one obtained in the previous section using the second moment method. Following the arguments in \eqref{eq:peps} -- \eqref{eq:sum_ui}, recall that we need to bound
\be
P( \mc{E}(\bfs) \mid \abs{\bfs}^2 = z^2 ) = P\left(\sum_{i=1}^{e^{nR}} U_i(\bfs) =0  \mid \abs{\bfs}^2 = z^2 \right).
\label{eq:sum_ui_new}
\ee

First, some definitions.
\begin{defi}
[Dependency Graphs  \cite{JansonBook}] Let $\{U_i\}_{i \in \mc{I}}$ be a family of random variables (defined on a common probability space). A dependency graph for $\{U_i\}$ is any graph $\Gamma$ with vertex set $V(\Gamma)= \mc{I}$ whose set of edges satisfies the following property: if $A$ and $B$ are two disjoint subsets of  $\mc{I}$ such that there are no edges with one vertex in $A$ and the other in $B$, then the families $\{U_i\}_{i \in A}$ and $\{U_i\}_{i \in B}$
are independent.
\end{defi}

\begin{fact}
\cite[Example $1.5$, p.11]{JansonBook}
Suppose $\{Y_\alpha\}_{\alpha \in \mc{A}}$ is a family of independent random variables, and each $U_i, i\in \mc{I}$ is a function of the variables
$\{Y_\alpha\}_{\alpha \in A_i}$ for some subset $A_i \subseteq \mc{A}$. Then the graph with vertex set $\mc{I}$ and edge set
$\{ij : A_i \cap A_j \neq \emptyset\}$ is a dependency graph for  $\{U_i\}_{i \in \mc{I}}$.
\label{fact:depgraph_ex}
\end{fact}

The graph $\Gamma$ with vertex set $V(\Gamma) = \{1,\ldots,e^{nR} \}$  and edge set $e(\Gamma)$ given by
\[  \{ij:  i\neq j \text{ and } \bfsh(i),\bfsh(j) \text { share at least one common term} \} \] is a dependency graph for the family
$\{ U_i(\bfs)\}_{i=1}^{e^{nR}}$, for each fixed $\bfs$. This follows from Fact \ref{fact:depgraph_ex} by recognizing that each  $U_i$ is a function of a subset of the columns of the matrix $\mathbf{A}$ and the columns of $\mathbf{A}$ are picked independently in the code construction.
\label{rem:depgraph}

For a given codeword $\bfsh(i)$, there are ${L \choose r} (M-1)^{L-r}$ other codewords that have exactly $r$ terms in common with $\bfsh(i)$, for
$0 \leq r \leq (L-1)$. Therefore each vertex in the dependency graph for the family $\{U_i(\bfs)\}_{i=1}^{e^{nR}}$ is connected to
\[ \sum_{r=1}^{L-1} {L \choose r} (M-1)^{L-r} = M^L - 1 - (M-1)^L \]
other vertices.

\begin{fact}[Suen's Inequality \cite{JansonBook}] Let $U_i \sim \text{Bern}(p_i),  i\in \mc{I}$, be a finite family of Bernoulli random variables having a dependency graph $\Gamma$.  Write $i \sim j$ if $ij$ is an edge in $\Gamma$. Define
\ben
\begin{split}
\lambda  = \sum_{i\in \mc{I}} \expec U_i,
 \  \, \Delta = \frac{1}{2} \sum_{i \in \mc{I}} \sum_{j \sim i} \expec(U_i U_j),
\  \, \delta =  \max_{i\in \mc{I}} \sum_{k \sim i}  \expec U_k.
\end{split}
\een
Then
\[ P\left(\sum_{i \in \mc{I}} U_i =0\right) \leq \exp\left(-\min \left\{ \frac{\lambda}{2}, \frac{\lambda}{6\delta}, \frac{\lambda^2}{8\Delta}  \right\} \right). \]
\end{fact}
We now apply Suen's inequality  with the dependency graph specified above for $\{U_i(\bfs)\}_{i=1}^{e^{nR}}$ to compute an upper bound for \eqref{eq:sum_ui_new}.

\textbf{First term $\lambda/2$}: Due to the symmetry of the codebook, $\expec(U_i (\bfs))$ does not depend on $i$. For any fixed $\bfs$ with $\abs{\bfs}^2=z^2$, we have
\be
\begin{split}
\lambda  = \sum_{i=1}^{e^{nR}} \expec(U_i(\bfs)) & = \ e^{nR}   P(U_1(\bfs) =1 \mid  \abs{\bfs}^2 = z^2) \\
 & \stackrel{(a)}{\geq} e^{nR} \cdot \frac{\kappa}{\sqrt{n}}e^{-nf(z^2)}
\end{split}
\label{eq:term1_bound}
\ee
where $(a)$ holds for $n$ sufficiently large with $f(z^2)$ given by  \eqref{eq:fz2_cond} and $\kappa>0$ a generic constant.
$(a)$ is a sharper version of \eqref{eq:EU1_ld}, obtained using the strong version of Cram{\'e}r's large-deviation theorem by Bahadur and Rao \cite{BahadurR60}.
We thus have a lower bound on $\lambda$ for sufficiently large $n$.

\textbf{Second term ${\lambda}/(6\delta)$}: Due to the symmetry of the code construction,
\be
\begin{split}
\delta &\triangleq \max_{i\in \{1, \ldots, e^{nR} \}} \sum_{k \sim i}  P\left(U_k(\bfs) =1 \mid  \abs{\bfs}^2 = z^2 \right)\\
&= \sum_{k \sim i}  P\left(U_k(\bfs) =1 \mid  \abs{\bfs}^2 = z^2 \right) \quad \forall i \in \{1, \ldots, e^{nR} \}\\
&= \sum_{r=1}^{L-1} {L \choose r} (M-1)^{L-r} \cdot P\left(U_1(\bfs) =1 \mid  \abs{\bfs}^2 = z^2 \right) \\
&= \left(M^L - 1 - (M-1)^L\right) P\left(U_1(\bfs) =1 \mid  \abs{\bfs}^2 = z^2 \right).
\end{split}
\ee
Using this together with the fact that \[ \lambda = M^L P(U_1(\bfs) =1 \mid  \abs{\bfs}^2 = z^2), \]
we have
\ben
\begin{split}
\frac{\lambda}{\delta}&  = \frac{M^L}{M^L - 1 - (M-1)^L} \stackrel{(a)}{=} \frac{1}{1- L^{-bL} - (1- L^{-b})^L}  \\
& = \frac{1}{1- L^{-bL} - \left[(1- L^{-b})^{L^b}\right]^{L^{1-b}}}
\end{split}
\een
where $(a)$ is obtained by substituting $M=L^b$.
Since \[ (1- L^{-b})^{L^b} \to e^{-1} \text{ as } L \to \infty \text{ for } b>1,\]   we have $(1- L^{-b})^{L^b} > e^{-2}$ for  sufficiently large $L$. Using this, we get
\be \label{eq:lambdel_step}
\frac{\lambda}{\delta} \geq \frac{1} {1 - L^{-bL} - (e^{-2})^{L^{1-b}} }.
\ee
From Taylor's theorem, the key term
\be
\begin{split}
& e^{-2L^{1-b}}   = e^{-2L^{-(b-1)}}   \\
 & > 1 - 2L^{-(b-1)}  + \frac{1}{4} \left(2L^{-(b-1)}\right)^2 \ \text{for large enough $L$}.
\end{split}
\ee
Using this  in \eqref{eq:lambdel_step}, we obtain that for $b>1$ and sufficiently large $L$
\be
\begin{split}
\frac{\lambda}{\delta}& \geq \frac{1}{2L^{-(b-1)} -  L^{-2(b-1)}  -  L^{-bL}   }
\geq \frac{L^{b-1}}{2}.
\end{split}
\label{eq:term2_bound}
\ee

\textbf{Third Term $\lambda^2/(8\Delta)$}: We have
\be
\begin{split}
& \Delta  = {\frac{1}{2} \sum_{i=1}^{M^L} \sum_{j \sim i} \expec\left[U_i(\bfs) U_j(\bfs) \mid \abs{\bfs}^2 = z^2\right] } \\
&  = \frac{M^L}{2} \sum_{r=1}^{L-1} {L \choose r} (M-1)^{L-r} \expec\left[U_1 U_j \mid \mathcal{F}_{1j}(r), \abs{\bfs}^2 = z^2 \right]
\end{split}
\label{eq:Delta_1}
\ee
where $\mathcal{F}_{ij}(r)$ denotes the event that $\bfsh(i),\bfsh(j)$ share $r$ common terms. The second equality above holds because of the symmetry of the code construction. Using \eqref{eq:Delta_1}, we have
\be
\begin{split}
& \frac{\lambda^2}{\Delta} = \frac{\left({M^L}  \expec(U_1(\bfs) \mid \abs{\bfs}^2 = z^2)\right)^2}
{\frac{M^L}{2} \sum_{r=1}^{L-1} {L \choose r} (M-1)^{L-r} \expec\left[U_1 U_j \mid \mathcal{F}_{1j}(r), \abs{\bfs}^2 = z^2 \right]}\\
&=\frac{2 M^L}{(M-1)^L}\\
& \quad  \cdot \frac{(M-1)^L  \left(\expec(U_1(\bfs) \mid \abs{\bfs}^2 = z^2)\right)^2}
{\sum_{r=1}^{L-1} {L \choose r} (M-1)^{L-r} \expec\left[U_1 U_j \mid \mathcal{F}_{1j}(r), \abs{\bfs}^2 = z^2 \right]}\\
& = \frac{2 M^L}{(M-1)^L} \cdot \frac{1}{T(z^2)}
\end{split}
\ee
where $T(z^2)$ was defined in \eqref{eq:tz2_def}. An upper bound for $T(z^2)$ was derived in the proof of Theorem \ref{thm:rd_sparc}. Using the bound given by \eqref{eq:tz2_ub} in Appendix \ref{app:lemma}, we obtain
\be
\frac{\lambda^2}{\Delta} \geq \kappa L^{(b-b_{min})( R - (1- D/\rho^2))/R}
\label{eq:term3_bound}
\ee
where $\kappa$ is a generic positive constant  and as specified in \eqref{eq:bmin},
\be b_{min} =  \frac{2.5R}{R - (1- D/\rho^2)}.  \label{eq:bmin_def}\ee

Using the bounds obtained in \eqref{eq:term1_bound}, \eqref{eq:term2_bound} and \eqref{eq:term3_bound} in Suen's inequality, we have for sufficiently large $n$
\be P\left(\sum_{i=1}^{e^{nR}} U_i(\bfs)=0 \mid \abs{\bfs}^2 = z^2 \right) \leq e^{-\min \{T_1,T_2,T_3 \}} \ee %
where
\be
\begin{split}
&T_1 >  e^{n\left( R-f(z^2)  - \frac{\log n}{2n} - \frac{\kappa}{n} \right)} \\
&T_2 > \kappa L^{b-1}, \\
&T_3> \kappa L^{ (b-b_{min}) (1- \frac{(1 - {D}/{\rho^2})}{R}) }.
\end{split}
\label{eq:T1T2T3_def}
\ee
Combining this with \eqref{eq:sum_ui_new} and \eqref{eq:err_bound}, we obtain
\be
\begin{split}
P_{e,n} & \leq  P(\abs{\bfs}^2 \geq \rho^2) \ +  \ \max_{z^2 \in (D, \rho^2)} P(\mc{E}(\bfs) \mid \abs{\bfs}^2 = z^2 ) \\
&< P(\abs{\bfs}^2 \geq \rho^2) \ + \ \max_{z^2 \in (D, \rho^2)} \exp(-\min \{T_1,T_2,T_3 \}) \\
& = P(\abs{\bfs}^2 \geq \rho^2) \ + \  \exp\Big(- \hspace{-3pt}\min_{z^2 \in (D, \rho^2)} \min \{T_1,T_2,T_3 \}\Big).
\label{eq:at_last}
\end{split}
\ee

We now show that if $b$ is chosen according to the statement of Theorem \ref{thm:err_exp}, the second  term in \eqref{eq:at_last}
decays exponentially in $n^t$, for some $t >1$.

First consider $T_1$. From \eqref{eq:T1T2T3_def}, we have
\be
\min_{z^2 \in (D, \rho^2)} T_1 >  \exp\left[n\left( R - \frac{\log n}{2n} - \frac{\kappa}{n} - \max_{z^2 \in (D, \rho^2)} f(z^2)  \right)\right]
\label{eq:min_T1}
\ee
In the proof of Lemma \ref{lem:EU1_asymp}, we showed that $f(z^2)$ is increasing in $(D, \rho^2)$ and so the maximum in this interval is
\[
f(\rho^2) = \frac{1}{2} \log \frac{\rho^2}{D}.
\]
Recall from \eqref{eq:R_a2} that $R = \max\{ \frac{1}{2} \log \tfrac{a^2}{D} ,  (1 - \tfrac{D}{a^2})\}$, and $a^2 > \rho^2$.
Therefore
\[
R > \frac{1}{2} \log\frac{\rho^2}{D} +  \frac{\log n}{2n} + \frac{\kappa}{n}
\]
for $n$ sufficiently large. Hence the right side of \eqref{eq:min_T1} grows exponentially  with $n$.

Next consider $T_2$. Since $b>2$, the lower bound $\kappa L^{b-1}$   given in \eqref{eq:T1T2T3_def} grows faster than $n$ as $n=\tfrac{b}{R} L\log L$.
Finally, the lower bound on $T_3$ given in in \eqref{eq:T1T2T3_def} will also grow faster than $n$ if
\[(b-b_{min}) \left(1- \frac{(1 - {D}/{\rho^2})}{R}\right) >1. \]
Substituting for $b_{min}$ from $\eqref{eq:bmin_def}$, we see that this is equivalent to
\be  b >   \frac{3.5R}{R - (1- D/\rho^2)}.
\label{eq:brho_cond} \ee
Therefore the lower bound on $P_{e,n}$ in \eqref{eq:at_last} can be written as
\[
P_{e,n} \leq P(\abs{\bfs}^2 \geq \rho^2) \ + \ \exp(-\Theta(n^t))
\]
where $t>1$ provided $b$ satisfies \eqref{eq:brho_cond}.  As $\rho^2 < a^2$, any $b$ chosen according to \eqref{eq:bmin_exp} in the statement of the theorem will satisfy \eqref{eq:brho_cond}. Since $\rho^2$ is an arbitrary number in the interval $(\sigma^2, a^2)$, setting $\rho^2=a^2 - \e$ completes the proof.

\section{Discussion} \label{sec:conc}
We have studied a new ensemble of codes for lossy compression 
where the codewords are structured linear combinations of elements of a dictionary. The size of the dictionary is a low-order polynomial in the block length. We showed that with minimum-distance encoding, this ensemble  achieves the optimal rate-distortion function of an i.i.d Gaussian source with the optimal error exponent for all distortions below $\frac{\sigma^2}{4.91}$, or equivalently for rates higher than $1.15$ bits per source sample. It was also shown that sparse regression codes are robust in the following sense: a SPARC designed to compress an i.i.d Gaussian source of variance $\sigma^2$  with distortion $D$ can compress any ergodic source of  variance less than $\sigma^2$ to within distortion $D$. Thus the sparse regression ensemble retains many of the good covering properties of the i.i.d random Gaussian ensemble.

An immediate goal is to prove that the optimal Gaussian rate-distortion function can be achieved for all values of target distortion with minimum-distance encoding. The  main challenge lies in controlling the asymptotic behavior of the function $T(.)$, defined in \eqref{eq:tz2_def}. If $X$ is the random variable denoting the number of codewords that represent the sequence $\bfs$ within distortion
$D$, it  can be verified that the ratio of the standard deviation and the mean of $X$ is governed by $\sqrt{T(z^2)}$ when $\abs{\bfs}^2=z^2$. The  proofs of both Theorems \ref{thm:rd_sparc} and \ref{thm:err_exp} require  $T(z^2)$ to go to zero with growing $n$, i.e.,  the standard deviation of $X$ must be  small compared to its mean. The condition under which  this happens is given by Lemma \ref{lem:Tz2}. When $D > \tfrac{\sigma^2}{4.91}$, the second moment method and Suen's inequality both fail for $\tfrac{1}{2} \log \frac{\sigma^2}{D} < R  < (1 - \tfrac{D}{\sigma^2})$ because $T(z^2) \to \infty$, i.e., the standard deviation of the number of solutions $X$ is much larger than its mean. One approach  to  improve the rate-distortion result for $D > \frac{\sigma^2}{4.91}$ is to expurgate atypical realizations of the design matrices from the SPARC ensemble in order to reduce the standard deviation of $X$; this would weaken the condition for $T(z^2)$ to go to $0$.

In this paper, the non-zero coefficients in each section of the codeword $\beta$ were all chosen to be equal. We can also choose the non-zero coefficients  to have different values in each section.\footnote{Recall that coefficient values are fixed a priori and revealed to both encoder and decoder; the codewords are determined only by the positions of the non-zero coefficient in each section.} This can help in designing fast encoding algorithms.  One such choice of varying section coefficients is used in \cite{RVGaussianFeasible} to derive a computationally efficient encoder based on successive approximation. This encoder chooses the codeword $\beta$ section by section, creating a residue after each step that is to be approximated by the subsequent sections.  This algorithm is shown to asymptotically  attain the optimal Gaussian distortion-rate function for \emph{all} rates.

The problem of compression with SPARCs is also  related to sparsity recovery in high-dimensional linear regression
\cite{WainSparse, AkTarokh, FletchRG}. While both problems aim to recover the positions of the non-zero coefficients of a sparse vector, the main difference is that both the positions and the values of the non-zero coefficients are unknown in the sparsity recovery problem, while the values are fixed a priori in the SPARC setting.  The connections between these two problems is an interesting topic for further investigation.
\appendices
\section{Proof of Lemma \ref{lem:Tz2}} \label{app:lemma}
Substituting $\alpha = r/L$ in \eqref{eq:tz2_def}, we can express $T(z^2)$  as follows.
\be
\begin{split}
& T(z^2) \\
 &=  \sum_{\alpha = 1/L}^{1-1/L} {L \choose L\alpha} (M-1)^{-L\alpha} \frac{\expec[U_1 U_j \mid \mathcal{F}_{1j}(\alpha), \abs{\bfs}^2 = z^2]}
{(\expec [ U_1 \mid \abs{\bfs}^2 = z^2])^2 } \\
& \leq (L-1) \max_{\alpha \in \{1/L, \ldots, (L-1)/L\}} \Bigg\{ {L \choose L\alpha} (M-1)^{-L\alpha}\\
&\hspace{1.3in} \cdot   \frac{\expec[U_1 U_j \mid \mathcal{F}_{1j}(\alpha), \abs{\bfs}^2 = z^2]}
{(\expec [ U_1 \mid \abs{\bfs}^2 = z^2])^2 } \Bigg\}
\end{split}
\ee
where $\mathcal{F}_{1j}(\alpha)$  denotes the event that $\bfsh(i),\bfsh(j)$ share $\alpha L$ common terms.
Taking logarithms, we obtain
\be
\begin{split}
 \log T(z^2) &  \leq  \max_{\alpha \in \{1/L, \ldots, (L-1)/L\}}  \Bigg\{\log (L-1) + \log {L \choose L\alpha}  \\
& - b \alpha L \log L +  \log {\expec[U_1 U_j \mid \mathcal{F}_{1j}(\alpha), \abs{\bfs}^2 = z^2]} \\
&   - 2 \log \expec [ U_1 \mid \abs{\bfs}^2 = z^2]\Bigg\}.
\end{split}
\label{eq:logTz2}
\ee
The asymptotic behavior of the last term above was established in the proof of Lemma \ref{lem:EU1_asymp}, and is given by \eqref{eq:EU1_ld}. The result in \eqref{eq:EU1_ld} can be sharpened using the strong version of Cram{\'e}r's large-deviation theorem by Bahadur and Rao \cite{BahadurR60}. Indeed, for all sufficiently large $n$
\be
\expec [ U_1 \mid \abs{\bfs}^2 = z^2] \geq \frac{\kappa}{\sqrt{n}}\exp(-nf(z^2))
\label{eq:EU1_bd}
\ee
where $f(z^2)$ is given by \eqref{eq:fz2_cond} and $\kappa>0$ is a constant. Here and in the sequel, $\kappa$ denotes a generic positive constant whose exact value is not needed.

Next we bound the term
$\expec[U_1 U_j \mid \mathcal{F}_{1j}(\alpha), \abs{\bfs}^2 = z^2]$.
We have
\be
\begin{split}
& \expec[  U_1 U_j \mid \mathcal{F}_{1j}(\alpha), \abs{\bfs}^2 = z^2] \\
& =P\left(U_1(\bfs)=1, U_j(\bfs)=1 \mid \mathcal{F}_{1j}(\alpha), \abs{\bfs}^2 = z^2 \right) \\
&= P\left(  \frac{1}{n} \sum_{k=1}^n (\hat{S}_k(1) - S_k)^2 \leq D, \right. \\
& \hspace{0.4in} \left. \frac{1}{n}\sum_{k=1}^n (\hat{S}_k(j) - S_k)^2 \leq D \mid \mathcal{F}_{1j}(\alpha), \abs{\bfs}^2 = z^2 \right) \\
&= P\left( \frac{1}{n} \sum_{k=1}^n (\hat{S}_k(1) - z)^2 \leq D, \right. \\
& \hspace{0.4in} \left. \frac{1}{n}\sum_{k=1}^n (\hat{S}_k(j) - z)^2 \leq D \mid \mathcal{F}_{1j}(\alpha) \right)
\end{split}
\ee
where the last equality is due to the fact that $(\bfsh(1),\bfsh(j))$ has the same joint distribution as $(\mathbf{O}\bfsh(i),\mathbf{O}\bfsh(j))$ for any orthogonal (rotation) matrix $\mathbf{O}$. The $(\hat{S}_k(1), \hat{S}_k(j))$ pairs are i.i.d across $k$, and each is  jointly Gaussian with zero-mean vector and covariance matrix
\be
K_{\alpha} = \gamma^2 \begin{bmatrix} 1 & \alpha \\ \alpha & 1 \end{bmatrix},
\label{eq:Kr} \ee
when $\bfsh(1), \bfsh(j)$ share $r= \alpha L$ common terms. Using a two-dimensional Chernoff bound, we have $\forall t <0$ and sufficiently large $n$
\be
\begin{split}
& \frac{1}{n} \log P  \left( \sum_{k=1}^n \frac{(\hat{S}_k(1) - z)^2}{n} \leq D,  \right. \\
& \quad \,  \left. \sum_{k=1}^n \frac{(\hat{S}_k(j) - z)^2}{n} \leq D \mid \mathcal{F}_{1j}(\alpha)\right) \leq \frac{\kappa}{ \sqrt{n}} \exp(-n C_\alpha(t))
\end{split}
\label{eq:C_ut}
\ee
where $\kappa >0$ is a constant and for $t<0$
\be
\begin{split}
C_\alpha(t)  &=  2tD - \log \expec \left( e^{t(\hat{S}(1) -z)^2 + t(\hat{S}(j) -z)^2} \right) \\
&=2tD  - \frac{ 2tz^2}{1-2\gamma^2 t (1+\alpha)} \\
& \ + \frac{1}{2} \log (1-4\gamma^2 t + 4 \gamma^4 t^2 (1-\alpha^2)).
\end{split}
\label{eq:CDef}
\ee
The optimal value of $t$ is the one that maximizes the right side of \eqref{eq:CDef}. Since  this optimal value cannot be expressed analytically, we choose $t=\frac{t_0}{1+\alpha}$ where $t_0 =\frac{1}{2\gamma^2}\left( 1 - \frac{2z^2}{\sqrt{(\gamma^4 + 4z^2D)} - \gamma^2} \right)$ is  optimal when $\alpha=0$.\footnote{Note that $t_0$ needs to be negative. This holds when $z^2+ \gamma^2 >D$, which is satisfied for all $\abs{\bfs}^2 =z^2 \geq D$.  This is true for all source sequences considered in the error analysis, as explained in the discussion following \eqref{eq:err_bound}.} With this choice, for all sufficiently large $n$  \eqref{eq:C_ut} yields
\be
\expec[U_1 U_j \mid \mathcal{F}_{1j}(\alpha), \abs{\bfs}^2 = z^2] \leq \frac{\kappa}{ \sqrt{n}} \exp\left(-n C_\alpha\left(\frac{t_0}{1+\alpha}\right)\right)
\label{eq:EU1Uj}
\ee
where
\be
\begin{split}
C_\alpha\left(\frac{t_0}{1+\alpha}\right) & =  \frac{1}{\gamma^2(1+\alpha)}\left( D + z^2 - \frac{2z^2D}{A(z^2)} - \frac{A(z^2)}{2} \right) \\
&\quad  + \frac{1}{2} \log \left( \frac{4z^2} {A(z^2)(1+\alpha)} \left(\alpha +  \frac{z^2(1-\alpha)}{A(z^2)}\right)  \right)
\end{split}
\label{eq:Calph_t}
\ee
where $A(z^2) = \sqrt{(\gamma^4 + 4z^2D)} - \gamma^2$.

Using \eqref{eq:EU1Uj} and \eqref{eq:EU1_bd} in \eqref{eq:logTz2}, for all sufficiently large $L$ we have
\be
\begin{split}
\log T(z^2) & \leq  \max_{\alpha \in \{1/L, \ldots, (L-1)/L\}} \left\{\log (L-1) + \log {L \choose L\alpha} \right. \\
 & \qquad  - b \alpha L \log L  + n(g(z^2) + \gamma_n) \Big\}
\end{split}
\label{eq:logTz2_a}
\ee
where
\be
\begin{split}
g(z^2) & = \frac{\alpha}{\gamma^2(1+\alpha)} \left( D+ z^2 - \frac{2z^2D}{A(z^2)} - \frac{A(z^2)}{2} \right) \\
& \quad -\frac{1}{2} \log \left(\frac{1-\alpha}{1+\alpha} + \frac{\alpha A(z^2)}{(1+\alpha)z^2} \right)
\end{split}
\label{eq:gz2}
\ee
and  $\gamma_n = \frac{\log (\kappa n)}{2n}$.

The function $g(z^2)$ is strictly increasing in the interval $(D,\rho^2)$. This is seen by noting that the derivative
\be
\begin{split}
 \frac{dg}{dz^2} & = \frac{\alpha}{\gamma^2 (1+\alpha)}\left( 1- \frac{2D}{\sqrt{\gamma^4 + 4 z^2 D}} \right) \\
 & \quad + \frac{\alpha (A(z^2))^2}{4z^2\sqrt{\gamma^4 + 4 z^2 D} ((1-\alpha)z^2 + \alpha A(z^2))}
\end{split}
\ee
is strictly positive for $z^2 \geq D$.  Therefore $g(z^2)$ can be upper bounded in \eqref{eq:logTz2_a} by its maximum value
\be g(\rho^2) = \frac{1}{2} \log \left( \frac{1+ \alpha}{1 - \alpha(1-\frac{2D}{\rho^2})} \right).  \label{eq:g_rho2}\ee

Substituting $n=\frac{b}{R} L \log L$ in \eqref{eq:logTz2_a} and dividing both sides  by $L \log L$, we obtain that for all $z^2 \in (D, \rho^2)$:
\begin{align}
& \frac{\log T(z^2)}{L \log L}  \leq \max_{\alpha \in \{1/L, \ldots, (L-1)/L\}}  \left\{\frac{1}{L} + \frac{\log {L \choose L\alpha}}{L \log L} - b \alpha \right. \nonumber \\
& \hspace{1.41in} \left.+ \, \frac{b g(\rho^2)}{R} + \frac{1}{2L} + \frac{\kappa}{L \log L} \right\} \\
&= \max_{\alpha \in \{1/L, \ldots, (L-1)/L\}}  \left\{\frac{3}{2L} + \frac{\log {L \choose L\alpha} + \kappa}{L \log L} \right.  \nonumber \\
& \hspace{1.5in} \left. - \frac{b}{R}(\alpha R - g(\rho^2)) \right\} \\
& \leq \max_{\alpha \in \{1/L, \ldots, (L-1)/L\}}  \left\{\frac{3}{2L}   + \min \left( \alpha, 1-\alpha, \frac{\log 2}{\log L} \right) \right. \nonumber \\
& \hspace{1.3in}  \left. + \frac{\kappa}{L \log L} - \frac{b}{R}(\alpha R - g(\rho^2)) \right\}
\label{eq:logTz2_b}
\end{align}
where \eqref{eq:logTz2_b} is obtained using the bound
\[ \log {L \choose L\alpha} <   \min \ \{ \alpha L \log L, \  (1-\alpha) L \log L, \ L \log 2 \}. \]
$T(z^2)$ will go to zero with growing $L$ if the right side of \eqref{eq:logTz2_b} is strictly negative as $L$ grows large. This will be true if for
sufficiently large $L$, the following two conditions hold \emph{for all} $\alpha \in \{1/L, \ldots, (L-1)/L \}$:
\begin{enumerate}
\item $\alpha R  - g(\rho^2) > 0$, and

\item $b > \frac{R}{\alpha R - g(\rho^2)} \left( \frac{3}{2L} + \min \left( \alpha, 1-\alpha, \frac{\log 2}{\log L} \right) + \frac{\kappa}{L \log L} \right)$.
\end{enumerate}

The first condition holds due to the following claim.

\vspace{10pt}
\emph{\textbf{Claim}}: For $R> R_{sp}(D)$, the function $h(\alpha) = \alpha R  - g(\rho^2)$ is strictly positive in the interval $[\tfrac{1}{L}, \tfrac{L-1}{L}]$. Further, for all sufficiently large $L$  its minimum in the interval is attained at $\alpha=1/L$ where it equals
\[ h( {1}/{L})  = \frac{1}{L} \left( R - (1- D/\rho^2)\right) + \frac{\kappa}{L^2}, \quad \kappa >0.\]
\emph{\textbf{Proof of Claim}}:
 We first note that $h(0)=0$ and $h(1) =R - \frac{1}{2} \log \frac{\rho^2}{D} > 0$. That $h$ is positive in $[\tfrac{1}{L},  \tfrac{L-1}{L}]$ is seen
 by combining two observations:
 \begin{enumerate}
 \item[(a)] $h'(0) = R - (1 - \frac{D}{\rho^2})$,  which is positive and hence $h$ is increasing at $\alpha=0$.
 \item[(b)] In the interval $(0,1)$,  $h$ has at most one local maximum  and no minima. Indeed, by solving $h'(\alpha)=0$  it can be verified that $h$ is  increasing in $(0,1)$ when $\rho^2/D \leq 4$; when $\rho^2/D > 4$, $h$ has  one maximum (and no minima) in $(0,1)$ and the maximum occurs at
 \ben
 \begin{split}
 \alpha^* & = \frac{D/\rho^2}{(1- {2D}/{\rho^2})} \Bigg(1  +  \\
   & \quad \left.  \left[1  +  \frac{\rho^4}{D^2} \left( 1- \frac{2D}{\rho^2}\right)\left(1- \frac{(1-D/\rho^2)}{R}\right)\right]^{\frac{1}{2}} \right).
    \end{split}
 \een
 \end{enumerate}
 For $\rho^2/D \leq 4$, the second part of the claim follows from (b). For $\rho^2/D > 4$, (a) and (b) imply that the minimum value  of $h$ in $[\tfrac{1}{L}, \tfrac{L-1}{L}]$ is attained either at either of the end points. Using a Taylor expansion for the function $h(\alpha)$ around the point $\alpha=0$, we can write
\be
\begin{split}
& h(\alpha) = \alpha R  - g(\rho^2) = \alpha(R- (1 - D/\rho^2)) \\
& \qquad \qquad  + \,  \alpha^2 \frac{(1 - D/\rho^2)(D/\rho^2 - \zeta (1 - 2D/\rho^2))}{(1+ \zeta)^2 \ (1- \zeta(1-2D/\rho^2))^2}.
\end{split}
\label{eq:tay1}
\ee
for some $\zeta \in (0, \alpha)$. Therefore for large enough $L$
 \be
 h( {1}/{L})  = \frac{1}{L} \left( R - (1- D/\rho^2)\right) + \frac{\kappa}{L^2}, \quad \kappa >0.
 \label{eq:tay2}
 \ee
Similarly, using a Taylor expansion for $h(\alpha)$ around $\alpha=1$, we get
\be
 h( 1 -\tfrac{1}{L} ) =  \left(R - \frac{1}{2} \log \frac{\rho^2}{D}\right)
 + \frac{1}{L}\left( R - \frac{1}{4}(\rho^2/{D} -1)  \right)  - \frac{\kappa'}{L^2}
 \label{eq:tay3}
\ee
for some $\kappa' >0$. Since  $R > \frac{1}{2} \log \frac{\sigma^2}{D}$, the minimum of $h$ is attained at $1/L$ for sufficiently large $L$.
\qedfilled
\vspace{20pt}

Recall that the condition $b$ needs to satisfy is
\be
b  > \hspace{-3pt} \max_{\alpha \in \{\frac{1}{L}, \ldots, \frac{L-1}{L} \}}
\frac{ R \left(\frac{3}{2L} + \min \left( \alpha, 1-\alpha, \frac{\log 2}{\log L} \right) + \frac{\kappa}{L \log L}\right)} {\alpha R  - g(\rho^2)}.
\label{eq:bcond_L1}
\ee
 From the claim and its proof (the Taylor expansions in \eqref{eq:tay1}-\eqref{eq:tay3}),  it is seen that the maximum in \eqref{eq:bcond_L1} is attained at $\alpha=1/L$ for $L$ sufficiently large. Substituting this value, we get
\be
b > \frac{2.5 R + \frac{\kappa}{\log L}}{R - (1- D/\rho^2) + \frac{\kappa}{L}}.
\ee
The minimum value of $b$, denoted $b_{min}$ is obtained by letting $L \to \infty$:
\be
b_{min} = \frac{2.5 R}{R - \left(1 - D/\rho^2 \right)}.
\label{eq:bmin}
\ee
For $b > b_{min}$, \eqref{eq:logTz2_b} implies that for $L$ sufficiently large
\be
\begin{split}
& \frac{\log T(z^2)}{L \log L}  \\
& \leq \max_{\alpha \in \{1/L, \ldots, (L-1)/L\}}  \frac{-(b-b_{min})(\alpha R - g(\rho^2))}{R} + \frac{\kappa}{L \log L} \\
& =  \frac{-(b-b_{min})\left( R - (1- D/\rho^2)\right)}{R L}  + \frac{\kappa}{L \log L}
\label{eq:logTz2_c}
\end{split}
\ee
where the equality is obtained using the claim above.
Therefore for $L$ sufficiently large,
\be
T(z^2) \leq \kappa L^{-{(b-b_{min})( R - (1- D/\rho^2))}/{R}}
\label{eq:tz2_ub}
\ee
which goes to zero as $L$ (or $n$) goes to $\infty$. This completes the proof of the lemma.

\section*{Acknowledgement}
The authors would like to thank A. Barron for many useful discussions regarding sparse regression codes. They would also like to thank the anonymous reviewers for their comments.

\IEEEtriggeratref{13}


\begin{thebibliography}{34}

\bibitem{AntonyML}
A.~Barron and A.~Joseph, ``Least squares superposition codes of moderate
  dictionary size are reliable at rates up to capacity,'' {\em IEEE Trans. Inf.
  Theory}, vol.~58, pp.~2541--2557, May 2012.

\bibitem{AntonyFast}
A.~Joseph and A.~R. Barron, ``Fast sparse superposition codes have near
  exponential error probability for ${R} < \mathcal{C}$,'' {\em IEEE Trans.
  Inf. Theory}, vol.~60, pp.~919--942, Feb. 2014.

\bibitem{MP93}
S.~Mallat and Z.~Zhang, ``Matching pursuits with time-frequency dictionaries,''
  {\em IEEE Trans. Signal Processing}, vol.~41, pp.~3397 --3415, Dec. 1993.

\bibitem{BarronCDD08}
A.~R. Barron, A.~Cohen, W.~Dahmen, and R.~A. DeVore, ``{Approximation and
  learning by greedy algorithms},'' {\em Annals of Statistics}, vol.~36,
  pp.~64--94, 2008.

\bibitem{TroppOMP}
J.~Tropp and A.~Gilbert, ``Signal recovery from random measurements via
  orthogonal matching pursuit,'' {\em IEEE Trans. Inf. Theory}, vol.~53,
  pp.~4655 --4666, Dec. 2007.

\bibitem{CandesTaoLP}
E.~Candes and T.~Tao, ``Decoding by linear programming,'' {\em IEEE Trans. Inf.
  Theory}, vol.~51, pp.~4203 -- 4215, Dec. 2005.

\bibitem{DonohoCS}
D.~Donoho, ``Compressed sensing,'' {\em IEEE Trans. Inf. Theory}, vol.~52,
  pp.~1289 --1306, April 2006.

\bibitem{RVGaussianFeasible}
R.~Venkataramanan, T.~Sarkar, and S.~Tatikonda, ``Lossy compression via sparse
  linear regression: Computationally efficient encoding and decoding,'' in {\em
  Proc. IEEE Int. Symp. Inf. Theory}, 2013.
\newblock To appear in \emph{IEEE Trans. Inf. Theory.}

\bibitem{CoverThomas}
T.~M. Cover and J.~A. Thomas, {\em Elements of Information Theory}.
\newblock John Wiley and Sons, Inc., 2006.

\bibitem{EyForney93}
M.~Eyuboglu and J.~Forney, G.D., ``Lattice and trellis quantization with
  lattice- and trellis-bounded codebooks-high-rate theory for memoryless
  sources,'' {\em IEEE Trans. Inf. Theory}, vol.~39, pp.~46 --59, Jan 1993.

\bibitem{Zamir02}
R.~Zamir, S.~Shamai, and U.~Erez, ``Nested linear/lattice codes for structured
  multiterminal binning,'' {\em IEEE Trans. Inf. Theory}, vol.~48, pp.~1250
  --1276, June 2002.

\bibitem{GuptaVerduWeiss}
A.~Gupta, S.~Verd{\` u}, and T.~Weissman, ``Rate-distortion in near-linear
  time,'' in {\em Proc. IEEE Int. Symp. on Inf. Theory}, 2008.

\bibitem{KontGioran}
I.~Kontoyiannis and C.~Gioran, ``Efficient random codebooks and databases for
  lossy compression in near-linear time,'' in {\em IEEE Inf. Theory Workshop on
  Networking and Inf. Theory}, pp.~236 --240, June 2009.

\bibitem{JalaliWeiss}
S.~Jalali and T.~Weissman, ``Rate-distortion via {M}arkov {C}hain {M}onte
  {C}arlo,'' in {\em Proc. IEEE Int. Symp. on Inf. Theory}, 2010.

\bibitem{GuptaVerdu09}
A.~Gupta and S.~Verd{\`u}, ``Nonlinear sparse-graph codes for lossy
  compression,'' {\em IEEE Trans. Inf. Theory}, vol.~55, pp.~1961 --1975, May
  2009.

\bibitem{WainManeva10}
M.~Wainwright, E.~Maneva, and E.~Martinian, ``Lossy source compression using
  low-density generator matrix codes: Analysis and algorithms,'' {\em IEEE
  Trans. Inf. Theory}, vol.~56, no.~3, pp.~1351 --1368, 2010.

\bibitem{polarrd}
S.~Korada and R.~Urbanke, ``Polar codes are optimal for lossy source coding,''
  {\em IEEE Trans. Inf. Theory}, vol.~56, pp.~1751 --1768, April 2010.

\bibitem{KontSPARC}
I.~Kontoyiannis, K.~Rad, and S.~Gitzenis, ``Sparse superposition codes for
  {G}aussian vector quantization,'' in {\em 2010 IEEE Inf. Theory Workshop},
  p.~1, Jan. 2010.

\bibitem{Lapidoth97}
A.~Lapidoth, ``On the role of mismatch in rate distortion theory,'' {\em IEEE
  Trans. Inf. Theory}, vol.~43, pp.~38 --47, Jan 1997.

\bibitem{SakMismatch1}
D.~Sakrison, ``The rate distortion function for a class of sources,'' {\em
  Information and Control}, vol.~15, no.~2, pp.~165 -- 195, 1969.

\bibitem{SakMismatch2}
D.~Sakrison, ``The rate of a class of random processes,'' {\em IEEE Trans. Inf.
  Theory}, vol.~16, pp.~10 -- 16, Jan 1970.

\bibitem{JansonBook}
S.~Janson, {\em Random Graphs}.
\newblock Wiley, 2000.

\bibitem{janson98}
S.~Janson, ``New versions of Suen's Correlation Inequality,'' {\em Random
  Structures and Algorithms}, vol.~13, no.~3-4, pp.~467--483, 1998.

\bibitem{SakrisonFin}
D.~Sakrison, ``A geometric treatment of the source encoding of a {G}aussian
  random variable,'' {\em IEEE Trans. Inf. Theory}, vol.~14, pp.~481 -- 486,
  May 1968.

\bibitem{IngberKochman}
A.~Ingber and Y.~Kochman, ``The dispersion of lossy source coding,'' in {\em
  Data Compression Conference (DCC)}, pp.~53 --62, March 2011.

\bibitem{KostinaV12}
V.~Kostina and S.~Verd{\'u}, ``Fixed-length lossy compression in the finite
  blocklength regime,'' {\em IEEE Trans. on Inf. Theory}, vol.~58, no.~6,
  pp.~3309--3338, 2012.

\bibitem{MartonRD74}
K.~Marton, ``Error exponent for source coding with a fidelity criterion,'' {\em
  IEEE Trans. Inf. Theory}, vol.~20, pp.~197 -- 199, Mar 1974.

\bibitem{IharaKubo00}
S.~Ihara and M.~Kubo, ``Error exponent for coding of memoryless {G}aussian
  sources with a fidelity criterion,'' {\em IEICE Trans. Fundamentals},
  vol.~E83-A, p.~1891–1897, Oct. 2000.

\bibitem{Den2008LD}
F.~Den~Hollander, {\em Large deviations}, vol.~14.
\newblock American Mathematical Society, 2008.

\bibitem{DemboZbook}
A.~Dembo and O.~Zeitouni, {\em Large Deviations Techniques and Applications}.
\newblock Springer, 1998.

\bibitem{BahadurR60}
R.~R. Bahadur and R.~R. Rao, ``On deviations of the sample mean,'' {\em The
  Annals of Mathematical Statistics}, vol.~31, no.~4, 1960.

\bibitem{WainSparse}
M.~J. Wainwright, ``Information-theoretic limits on sparsity recovery in the
  high-dimensional and noisy setting,'' {\em IEEE Trans. Inf. Theory}, vol.~55,
  no.~12, pp.~5728--5741, 2009.

\bibitem{AkTarokh}
M.~Akcakaya and V.~Tarokh, ``Shannon-theoretic limits on noisy compressive
  sampling,'' {\em IEEE Trans. Inf. Theory}, vol.~56, pp.~492 --504, Jan. 2010.

\bibitem{FletchRG}
A.~K. Fletcher, S.~Rangan, and V.~K. Goyal, ``Necessary and sufficient
  conditions for sparsity pattern recovery,'' {\em IEEE Trans. Inf. Theory},
  vol.~55, pp.~5758--5772, Dec. 2009.

\end{thebibliography}
\end{document}